\documentclass[a4, 9pt]{article}
\makeatletter
\usepackage[margin=1in]{geometry}
\usepackage{amsmath}
\usepackage{amssymb}
\usepackage{float}
\usepackage{psfrag}
\usepackage[dvips]{graphicx,color}
\usepackage{cite}
\usepackage{authblk}
\newtheorem{theorem}{Theorem}[section]
\newtheorem{lemma}[theorem]{Lemma}
\newtheorem{proposition}[theorem]{Proposition}

\newtheorem{definition}{Definition}
\newenvironment{proof}[1][Proof]{\begin{trivlist}
\item[\hskip \labelsep {\bfseries #1}]}{\end{trivlist}}

\newcommand{\BigO}[1]{\ensuremath{\operatorname{O}\left(#1\right)}}
\newcommand{\qed}{\nobreak \ifvmode \relax \else
      \ifdim\lastskip<1.5em \hskip-\lastskip
      \hskip1.em plus0em minus0.5em \fi \nobreak
      \vrule height0.45em width0.65em depth0.25em\fi}

\title{Maximizing Utility Among Selfish Users in Social Groups}
\author[1]{Ashwin Pananjady\thanks{ee10b025@ee.iitm.ac.in}}
\author[1]{Vivek Kumar Bagaria\thanks{ee10b047@ee.iitm.ac.in}}
\author[2]{Rahul Vaze\thanks{vaze@tcs.tifr.res.in}}
\affil[1]{Department of Electrical Engineering, Indian Institute of Technology Madras}
\affil[2]{School of Technology and Computer Science, Tata Institute of Fundamental Research}

\begin{document}
%
%
\maketitle
\begin{abstract}
We consider the problem of a social group of users trying to obtain a ``universe" of files, first from a server and then via exchange amongst themselves. We consider the selfish file-exchange paradigm of give-and-take, whereby two users can exchange files only if each has something unique to offer the other. We are interested in maximizing the number of users who can obtain the universe through a schedule of file-exchanges. We first present a practical paradigm of file acquisition. We then present an algorithm which ensures that at least half the users obtain the universe with high probability for $n$ files and $m=\BigO{\log n}$ users when $n\rightarrow\infty$, thereby showing an approximation ratio of $2$. Extending these ideas, we show a $1+\epsilon_1$ - approximation algorithm for $m=\BigO{n}$, $\epsilon_1>0$ and a $(1+z)/2 +\epsilon_2$ - approximation algorithm for $m=\BigO{n^z}$, $z>1$, $\epsilon_2>0$. Finally, we show that for any $m=\BigO{e^{o(n)}}$, there exists a schedule of file exchanges which ensures that at least half the users obtain the universe.
\end{abstract}
\section{Introduction}
\label{sec:intro}  
We consider a peer-to-peer (P2P) network scenario, where each node (or user) desires a certain fixed set of files residing on a central server, e.g. a full high-definition movie \cite{knoke2008social}. For each node, the cost associated with downloading all the files from the server is prohibitively large. However, the cost of exchanging files across the nodes is very low. The high cost of file acquisition from the server could be due to large delay, license fee, power etc., while the low cost of P2P exchange could be due to nodes sharing bandwidth, physical proximity between the nodes, mutual cooperation, etc. Popular applications based on this concept include web caching, content distribution networks, P2P networks, etc. Device-to-device (D2D) communication in cellular wireless networks is another example, where mobile phones cooperate with one another to facilitate communication.

Clearly, to keep cost low, it is intuitive for each node to download a small fraction of files from the server, and then obtain the remaining files through exchanges with other nodes. Free riding -the phenomenon in which a node does not download anything and still obtains all the files from other nodes - is possible in P2P networks, and can lead to attacks like whitewashing, collusion, fake services, Sybil attack \cite{karakaya2009free,fox2001peer,dinger2006defending}, etc. To avoid free riding \cite{feldman2006free,feldman2005overcoming,locher2006free,rahman2010improving,nishida2010global,karakaya2009free,mol2008give}, the file exchange among nodes follows a \emph{give-and-take} criterion. Two nodes can exchange files if both have some file(s) to offer each other.

The problem we consider in this paper is to find an algorithm for scheduling file exchanges between nodes such that at the end of algorithm, when no more exchanges are possible due to the condition imposed by the give-and-take criterion, the number of nodes that receive all files is maximized. Prior work \cite{aggarwal2013social} has considered maximizing the aggregate number of files received by all nodes. Depending on the utility of the user, e.g. watching a movie, maximizing the number of nodes that obtain all the files could be more important than maximizing the total number of files across the network.

Finding an optimal algorithm for scheduling file exchanges with the give-and-take model is challenging, since at each step there are many pairs of nodes that satisfy the give-and-take criterion and the choice of exchange at each step determines the final outcome. As a consequence, the number of feasible schedules is exponential. We therefore concern ourselves with providing approximation algorithms for the problem that have polynomial time complexity. An algorithm to solve a maximization problem is said to have an approximation ratio $\rho > 1$ if it always returns a solution greater than $\frac{1}{\rho}$ times the optimal solution. We call such an algorithm a $\rho$ approximation algorithm, or simply a $\rho$ algorithm.

We deal with the following file acquisition paradigm. Let each user download each file from the server independently with probability $p$. Let there be a total of $n$ files and $m$ users. Note that the download cost for each user is directly proportional to $p$, since the cost follows a binomial distribution. Users can either choose $p$ before entering the system (a-priori) or after (a-posteriori), in which case it could depend on the number of users or files. In this paper, we propose a deterministic iterative tree-splitting algorithm with polynomial complexity. The algorithm involves recursively dividing the users into two equally-sized groups, thereby forming a tree. The divisions are defined appropriately, keeping the give-and-take criterion in mind. Exchanges are then effected from the leaves up to the root of the tree.

Our contributions for various regimes of $m$ and $n$ are listed below. All of these hold with very high probability as $n\rightarrow\infty$:
\begin{itemize}
\item For $m=\BigO{\log n}$, our algorithm has an approximation ratio $\rho =2$, provided $p$ can be chosen a-posteriori, so the cost to each user is determined by the number of users and files.
\item When $m=\BigO{n}$, we present an algorithm that has an approximation ratio of $\rho =2$ when $p$ is chosen a-priori, which implies that at least half the users can obtain the universe at a vanishingly small cost to themselves (by minimizing $p$), and an approximation ratio of $\rho =1+\epsilon_1$ when $p$ can be chosen a-posteriori $(\epsilon_1>0)$, in which case the cost is dependent on the number of users and files.
\item When $m = \BigO{n^z}$ for some $z>1$, we present an algorithm that has an approximation ratio of $\rho=1+z$ when $p$ is chosen a-priori, and an approximation ratio of $\rho=\frac{1+z}{2}+\epsilon_2$ when $p$ can be chosen a-posteriori $(\epsilon_2>0)$.
\item For any $m = \BigO{e^{o(n)}}$, we show that there exists a schedule of file exchanges that ensures that at least half the users obtain the universe of files.
\end{itemize}

\section{Preliminaries}
\subsection{Notation} \label{notation}
\begin{enumerate}
\item $T$: set of files on the server (the complete universe); $|T|=n$.
\item $U = \{u_1, u_2, \ldots, u_m\}$: set of all users.
\item $C_i^t$: set of files that user $u_i$ possesses at time step $t$. $C_i^0$ represents the files that user $u_i$ possesses at the start, i.e. before any exchanges.
\item $F = \bigcup_{i}C_i^0=\{1,2,3,\ldots, f\}$: set of all files in the achievable universe.
\item $u_i\leftrightarrow u_j$: exchange between users $u_i$ and $u_j$.
\item $F_{G}$: set of files that a group of users $G$ possesses. $F_{G}=\bigcup_{i:u_i\in G}C_i^t$ at time $t$. For example, $F_U=F$.
\end{enumerate}

\subsection{Problem Definition} \label{probdef}
We assume that a single central server contains $n$ files represented by the set $T$. Let $m$ users represented by the set $U = \{u_1, u_2, \ldots, u_m\}$ initially obtain some files from this server, but at high cost to themselves. Let the files obtained by a user $u_i$ be represented by $C_i^0$. We consider a particular file acquisition paradigm, which we state as follows.
\begin{definition}[Random Sampling] \label{randomparadigm}
Each user picks up a file from the server with a probability $p$ (i.e. a file is not picked up by that user with probability $q=1-p$). The pickup of files is i.i.d across files and users.
\end{definition}

The parameters of the paradigm, i.e. the value of $p$ (or equivalently, $q$) can either be decided a-priori - when the users are oblivious to the number of other users in the social group and so decide the value of $p$ before entering the system - or a-posteriori - when the users appropriately choose the value of $p$ depending on $m$ and $n$. We will consider both cases.

Let $F=\bigcup_iC_i^0=\{1,2,3\ldots,f\}$ be the \emph{achievable universe}. The primary objective for each user is to obtain all $f$ files in the achievable universe $F$, since these represent all the files obtainable through exchanges. To this end, the users are interested in disseminating the files among themselves at low cost. However, since each user is selfish, file transfer between users can only occur via the give-and-take protocol, which we explain using the following definitions:
\begin{definition}[GT Criterion \cite{aggarwal2013social}] \label{GT}
Two users $u_i$ and $u_j$ are said to satisfy the GT criterion at time $t$ if $C_i^t\not\subseteq C_j^t$ and $C_j^t\not\subseteq C_i^t$. 
\end{definition}
\begin{definition}[Exchange: $u_i\leftrightarrow u_j$] \label{exchange}
Two users can only exchange files if they satisfy the GT criterion. After an exchange between users $u_i$ and $u_j$ (at time $t$, say), $C_i^{t+1}=C_j^{t+1}=C_i^{t}\bigcup C_j^{t}$.
\end{definition}
Def. \ref{GT}\& \ref{exchange} form the cornerstone of the give-and-take protocol.

We call a sequence of exchanges between users a \emph{schedule}, and represent it by $\mathcal{S}$. The set of all schedules is denoted by $\mathcal{X}$. It is clear that any schedule will lead to a situation in which no further exchanges are possible. At this stage, we represent the set of \emph{satisfied} users who have obtained all $f$ files by $U_{sat}$.

Formally, we are interested in the following problem:
\begin{align} 
\max |U_{sat}| \nonumber \\ 
\text{Subject to: } \mathcal{S}\in \mathcal{X} \label{objfunc}
\end{align}

As mentioned in Section \ref{sec:intro}, there are an exponential number of schedules in $\mathcal{X}$, and our interest lies in approximating \eqref{objfunc} in polynomial time.

\section{Recursive Algorithm} \label{recursive}
To solve \eqref{objfunc}, we present a recursive algorithm, which we call the TreeSplit algorithm, which repeatedly divides the users into two groups of equal size and then effects exchanges appropriately. It is illustrated in Fig. \ref{treefig}.
\begin{figure}
    \begin{center}
    \psfrag{B}[][t][.7]{$U$}
    \psfrag{C}[][l][.7]{$m$ users}
    \psfrag{E}[][][.7]{$G_1^1$}
    \psfrag{F}[][][.7]{$G_2^1$}
    \psfrag{A}[][r][.7][90]{$k=\log m $ levels}
    \psfrag{G}[][l][.7]{$m/2$ users}
    \psfrag{D}[][b][.7]{Level $1$ split}
    \psfrag{H}[][][.7]{Level $k$ split}
    \psfrag{I}[][][.7]{$m$ leaves}
    \includegraphics[scale=.4]{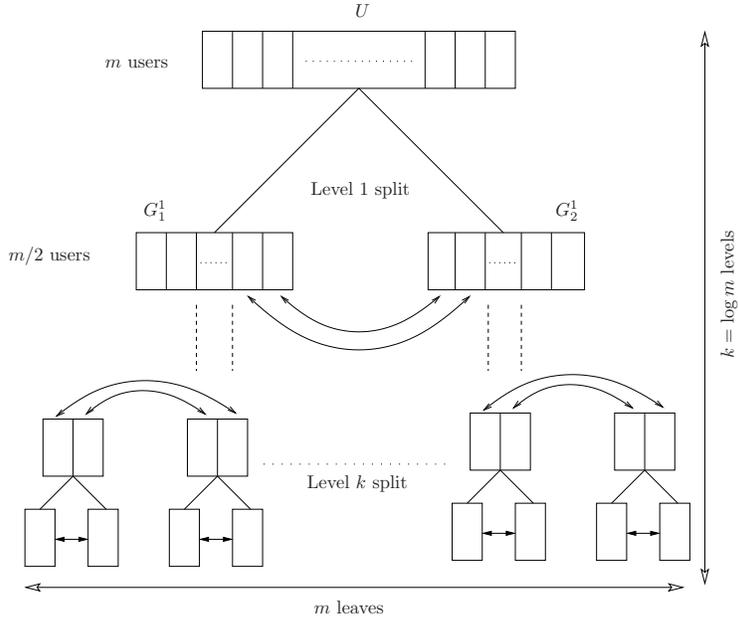}
    \caption{TreeSplit of users. $\leftrightarrow$ represents exchange, solid lines arising out of a group represent the division of that group which obeys the splitting condition.} \label{treefig}
    \end{center}
\end{figure}
Recall from Section \ref{notation} that the set of files contained by a group of users $G$ is represented by $F_{G}$. Let there be $m=2^k$ users for some $k$. Divide the $m$ users into two groups $G_1^1$ and $G_2^1$ of $m/2$ users each, satisfying the \emph{splitting condition}, defined as $F_{G_1^1}\not\subseteq F_{G_2^1}$ and $F_{G_2^1}\not\subseteq F_{G_1^1}$. We call this division a \emph{level 1} split, which is indicated by the superscript of groups $G_i^1$. If we can design a schedule to ensure a scenario in which all users in $G_1^1$ contain $F_{G_1^1}$, and all users in $G_2^1$ contain the files $F_{G_2^1}$, we can then effect pairwise exchanges between users in $G_1^1$, i.e., $u_{iG_1^1}, i=\{1,2\ldots,m/2\}$, and users in $G_2^1$, i.e., $u_{iG_2^1}, i=\{1,2\ldots,m/2\}$. In other words,
\begin{align*}
u_{iG_1}\leftrightarrow u_{iG_2} \text{  } \forall \text{ }i\in\{1,2,\ldots, m/2\} \nonumber  
\end{align*}
ensuring that all $m$ users obtain all files in $F$, maximizing \eqref{objfunc}.

Thus, essentially, we have reduced the problem to equivalent problems on two smaller groups of users $G_1^1$ and $G_2^1$, separately. Taking this recursion forward, we need to see what happens at the base case (at a level $k$ split), when a group of two users $G_j^{k-1}$ is divided into the groups $G_i^{k}$ and $G_{i+1}^{k}$ of one user each, for some $i,j$. If the division of $G_j^{k-1}$ obeys the splitting condition, then $F_{G_i^{k}}\not\subseteq F_{G_{i+1}^{k}}$ and $F_{G_{i+1}^{k}}\not\subseteq F_{G_{i}^{k}}$. Since both $G_i^k$ and $G_{i+1}^k$ consist of one user each, the splitting condition forces the GT criterion to be satisfied between the two users $G_i^k$ and $G_{i+1}^k$. So we can effect an exchange between users $G_i^k$ and $G_{i+1}^k$ and ensure that $G_j^{k-1}$ contains two users who possess all the files in $F_{G_j^{k-1}}$. Pairwise exchanges as we move up the tree will therefore solve problem \ref{objfunc}, and the result is summarised in the following Theorem.

\begin{theorem}
For $m=2^k$ for some $k$, we can use the TreeSplit algorithm to schedule exchanges such that all users obtain the entire achievable universe $F$, provided that divisions obeying the splitting condition can be enforced at all levels. \label{splitcond}
\end{theorem}

Notice that the TreeSplit algorithm is valid only if $m$ is a power of $2$. We now look at what happens for a general $m$.

We define the \emph{UniquePick} algorithm, which we apply on the $m$ users as follows: Initially, let $U_{suff}=U$. We now \emph{prune} $U_{suff}$ as follows: Pick any user in $U_{suff}$. If he contains a \emph{unique} file that no other user in $U_{suff}$ contains, keep him in $U_{suff}$; otherwise throw him out. Do this in any arbitrary order until all users have been checked for uniqueness. 
\begin{proposition}
The UniquePick algorithm ensures that $U_{suff}$ finally contains all files in the achievable universe, i.e. $F_{U_{suff}}=F$. \label{uniquepick} 
\end{proposition}
\begin{proof}
The UniquePick algorithm removes only those users who do not have any unique files. Therefore, even if a user $u_i$ is removed from $U_{suff}$, every file $j\in C_i^0$ is present in $U_{suff}\backslash u_i$. Since $U_{suff}=U$ initially, and no file of $F$ is removed from all users in $U_{suff}$ up to the end of the Unique-Pick algorithm, $U_{suff}$ must finally contain all files in $F$.
\qed
\end{proof}

\begin{theorem}
For any $m$, provided that the splitting condition can be enforced in all divisions, the TreeSplit algorithm has a 2-approximation algorithm to \eqref{objfunc}. \label{2approx}
\end{theorem}
\begin{proof}
We consider the power of $2$ closest to but less than $m$, and call this $2^y$ ($y=\lfloor\log_2m\rfloor$). Note that $2^y/m>1/2$, by definition.
At the end of the UniquePick algorithm, either \textbf{(i)} $|U_{suff}|\leq 2^y$ or \textbf{(ii)} $|U_{suff}|> 2^y$. In the case of condition \textbf{(i)}, we add some users from $U\backslash U_{suff}$ to $U_{suff}$ so that $|U_{suff}|= 2^y$. We know from Prop. \ref{uniquepick} that $2^y$ users contain $F$, so we can now use Thm. \ref{splitcond} to ensure that $2^y$ users obtain $F$ provided the splitting condition can be enforced in all divisions. So, at least half the users obtain the achievable universe (since $2^y/m>1/2$), giving us an approximation ratio of $2$.

In the case of condition \textbf{(ii)}, we know that more than $2^y$ users are such that each contains a unique file, and $F_{U_{suff}}=F$ (by Prop. \ref{uniquepick}). We can therefore use the Polygon Algorithm from \cite{aggarwal2013social} to schedule exchanges and ensure that all users in $U_{suff}$ obtain the universe. Again, more than $2^y$ users obtain $F$, giving us an approximation ratio of $2$. 
\qed
\end{proof}

One important assumption we made in the TreeSplit algorithm was that the splitting condition can be enforced in every division; we did not explain how the division itself is done. In the sections that follow, we will consider different regimes of $m$ and $n$ for which the splitting condition is satisfied with high probability.

\subsection{Approximation ratio of $2$ when $m=\BigO{\log n}$, $n\rightarrow\infty$} \label{lognsection}
When $m=\BigO{\log n}$, we apply the TreeSplit algorithm to the users. We \textbf{randomly divide} groups as we go down the tree. We show that all divisions done randomly satisfy the splitting condition with high probability.

Recall the Random Sampling paradigm, in which a file is not chosen by a user with probability $q$. Let $m=2^k$. If all divisions satisfy the splitting condition in this case, then the approximation ratio is $2$ by Thm. \ref{2approx}. The TreeSplit algorithm is applied on these $m$ users. During the algorithm, a total of $1+2+2^2+\ldots+2^{k-1}=2^k-1=m-1$ groups are divided ($2^j$ groups at level $j$), as is obvious from Fig. \ref{treefig}. We represent these $m-1$ groups by the set $\mathcal{K}$. We would like the splitting condition to be obeyed whenever any group $M\in \mathcal{K}$ is divided. To this end, we define an event - that the splitting condition is violated during the division of a group $M\in \mathcal{K}$ - and denote it by $A_M$. We state the following Lemma. 
\begin{lemma}
The probability of occurrence of $A_M$ when the group $M$ under consideration has some $d$ users is
\begin{align}
\mathcal{P}(A_M)&= 2\sum_{\ell=0}^{n} \dbinom{n}{\ell} (1-q^{d/2})^\ell(q^{d/2})^{n-\ell}(1-q^{d/2})^\ell. \label{probexp}
\end{align}
\end{lemma}
\begin{proof}
We know that group $M$ has $d$ users. Now consider a division of these $d$ users into two groups $M_1$ and $M_2$ of $d/2$ users each. If $M_1$ and $M_2$ are such that $F_{M_1}\subseteq F_{M_2}$ or $F_{M_2}\subseteq F_{M_1}$, then the splitting condition is violated and event $A_M$ occurs. Let us assume that $F_{M_1}\subseteq F_{M_2}$. Note that $(1-q^{d/2})^\ell(q^{d/2})^{n-\ell}$ is the probability that group $M_1$ has only some $\ell$ files and $\dbinom{n}{\ell}$ is the number of ways of choosing these $\ell$ files from $n$ files in $T$. Also, $(1-q^{d/2})^\ell$ is the probability that group $M_2$ also contains those $\ell$ files (making $F_{M_2}$ a superset). Summing over all possible $\ell$, and multiplying by $2$ to factor in the equivalent case in which $F_{M_2}\subseteq F_{M_1}$, we get the expression in \eqref{probexp}.
\qed
\end{proof}
Note that \eqref{probexp} can also be written as:
\begin{equation}
\mathcal{P}(A_M)= 2((1-q^{d/2})^2+q^{d/2})^n= 2(1+q^{d}-q^{d/2})^n. \label{probexp1}
\end{equation}
As stated earlier, we want to ensure that the splitting condition is met for all $M\in\mathcal{K}$. When the splitting condition is not met for any group in $\mathcal{K}$, we say that an \emph{error} has occurred. Note that the probability of error is the union of the probabilities of all $A_M$, and can be written as follows:
\begin{align}
\mathcal{P}(error)=\mathcal{P}(\bigcup_{M\in \mathcal{K}}A_M)&\leq \sum_{M\in \mathcal{K}}\mathcal{P}(A_M), \nonumber\\
					  &= 2\sum_{M\in \mathcal{K}}(1+q^{|M|}-q^{|M|/2})^n, \label{bound1}
\end{align}
where $|M|$ represents the number of users in group $M$. Since the number of terms in the RHS of \ref{bound1} is $|\mathcal{K}|=m-1<m$ and $\max_{M\in\mathcal{K}}(1+q^{|M|}-q^{|M|/2})^n$ occurs at $|M|=m$ for a sufficiently large $m$, we get
\begin{equation}					  
\mathcal{P}(error)\leq 2m (1+q^{m}-q^{m/2})^n. \label{probbound}
\end{equation}

We now enforce a constraint on file acquisition. Note that until now, we have only been concerned with all users obtaining the achievable universe $F$. But in practice, each user desires the complete universe $T$. Before discussing how each user can obtain $T$, we define a \emph{file cover}.
\begin{definition}[File Cover]
A group of users $G$ is said to be a file cover if $F_G=T$, the universe of available files.
\end{definition}
\begin{lemma}
A group of $m$ users is a file cover if $nq^m\rightarrow 0$ as $n\rightarrow\infty$.
\end{lemma}
\begin{proof}
We consider a suitably defined bad event, that the group of $m$ users $U$ is not a file cover ($F_U\neq T$). We call this event $D$, and note that
\begin{equation}
\mathcal{P}(D)= 1-(1-q^{m})^n\leq nq^m \label{probd}
\end{equation}
It is clear from \eqref{probd} that $\mathcal{P}(D)\rightarrow 0$ as $n\rightarrow\infty$ if $nq^m\rightarrow 0$, so the set $U$ is a file cover with high probability.
\end{proof}

\begin{lemma}
If $m=c\log n$, the TreeSplit algorithm has an approximation ratio of $2$ to \eqref{objfunc} provided that the value of $q$ (or $p$) can be chosen a-posteriori such that $-2/c<\log q<-1/c$.
\end{lemma}
\begin{proof}
Note that for $m=c\log n$, using \eqref{probbound} and the identity that $1-x<\exp(-x)$,
\begin{align}
\mathcal{P}(error)&< 2c\log n \exp\Big({n(q^{c\log n} -q^{\frac{c\log n}{2}})}\Big), \nonumber \\ 
		&= 2c\log n \exp\Big({(n^{\frac{c\log q}{2}} -1)\cdot n^{(1+\frac{c\log q}{2})}}\Big). \label{bound}
\end{align}
The RHS of \eqref{bound} goes to zero as $n\rightarrow\infty$ when $\log q>-2/c$, so all random divisions in the TreeSplit algorithm obey the splitting condition if $\log q>-2/c$. Furthermore, from \eqref{probd}, we see that for this case,
\begin{equation}
\mathcal{P}(D)\leq nq^{c\log n} \label{probdlog}
\end{equation}
$\mathcal{P}(D)\rightarrow 0$ as $n\rightarrow\infty$ provided $\log q<-1/c$. Therefore, for $-2/c<\log q<-1/c$, the TreeSplit algorithm ensures that all users obtain the complete universe $T$ if $m$ is a power of $2$, which by Thm. \ref{2approx}, is a 2-approximation for general $m$.
\qed
\end{proof}

We now look at a more general regime, in which the number of users grows linearly with the number of files.
\subsection{Approximation ratio of $1+\epsilon$ for $m=\BigO{n}$} \label{nsection}
For $m=\BigO{n}$, we propose the \emph{Partition+TreeSplit} algorithm. There are two steps to the algorithm. In step \textbf{(i)}, the users are randomly partitioned into a set of groups $\mathcal{G}=\{G_1,G_2,\ldots\}$ where each group $G_i\in\mathcal{G}$ has $w\log n$ users. We would like for all groups $G_i\in\mathcal{G}$ to be file covers. In step \textbf{(ii)}, the TreeSplit algorithm with random divisions is applied to some $v\log n$ users in every group $G_i\in\mathcal{G}$, where $v\log n$ is a power of $2$ closest to but less than $w\log n$. Let $m=\alpha n$. The Partition+TreeSplit algorithm is illustrated in Fig. \ref{treecover}. 

\begin{lemma}
If $w$, $v$ and $q$ are such that $-2/v<\log q<-2/w$ and $v\log n$ is a power of $2$ closest to but less than $w\log n$, then the Partition+TreeSplit algorithm ensures that at least $\frac{\alpha n}{w\log n}v\log n$ users obtain the complete universe $T$. \label{splittree}
\end{lemma}
\begin{proof}
In the Partition+TreeSplit algorithm, we would like to avoid errors at both steps of the algorithm. An error $E_1$ in step \textbf{(i)} occurs if any group $G_i\in\mathcal{G}$ is not a file cover. An error $E_2$ in step \textbf{(ii)} occurs if any division in the TreeSplit algorithm does not obey the splitting condition. We are therefore interested in analyzing the probability $\mathcal{P}(E_1\cup E_2)$, which can be written as:
\begin{align}
&\sum_{G_i\in\mathcal{G}}\mathcal{P}(G_i\text{ is not a file cover}) +
\sum_{G_i\in\mathcal{G}}\mathcal{P}(\text{Divisions among }v\log n \text{ users in }G_i\text{ do not obey the}\text{ splitting condition}). \label{totalprob}
\end{align}
Each $G_i\in\mathcal{G}$ has $w\log n=\BigO{\log n}$ users, for which we have already computed the probability that $G_i$ is not a file cover in \eqref{probdlog}. Similarly, since the term in the second summation involves $v\log n=\BigO{\log n}$ users, we have computed the probability of an error in division in \eqref{bound}. Also note that each summation in \eqref{totalprob} has $\alpha n/w\log n$ terms, since there are $\alpha n/w\log n$ groups in $\mathcal{G}$. Using \eqref{bound} and \eqref{probdlog}, we see that
\begin{equation}
\mathcal{P}(E_1\cup E_2)<\frac{\alpha n}{w\log n}\bigg(nq^{w\log n} + 2v\log n e^{(n^{\frac{v\log q}{2}} -1)\cdot n^{(1+\frac{v\log q}{2})}}\bigg). \label{longexp}
\end{equation}
We see that the first term of \eqref{longexp} goes to zero when $\log q<-2/w$ and the second goes to zero when $\log q >-2/v$, so the probability of error in the Partition+TreeSplit algorithm goes to zero for $q$ such that $-2/v<\log q<-2/w$.

We now look at how many users can obtain the universe via the Partition+TreeSplit algorithm. We know from \eqref{longexp} that if $q$ is such that $\log q<-2/w$, the random partition of $m$ users into the set of groups $\mathcal{G}$ ensures that all groups in $\mathcal{G}$ are file covers with high probability. Each group has $w\log n$ users and $|\mathcal{G}|$ is equal to $\alpha n/w\log n$. Let us now try to prune each group $G_i\in\mathcal{G}$ through the UniquePick algorithm to obtain $H_i$. $H_i$ can be of $2$ types: \textbf{(a)} $|H_i|> v\log n$ or \textbf{(b)} $|H_i|\leq v\log n$. 

If $H_i$ is of type \textbf{(a)}, we know by the definition of the UniquePick algorithm that it contains more than $v\log n$ users who contain unique files. So we can apply the Polygon algorithm of \cite{aggarwal2013social} to ensure that more than $v\log n$ users in $H_i$ obtain the complete universe $T$. 
If $H_i$ is of type \textbf{(b)}, add users from $G_i\backslash H_i$ so that $|H_i|=v\log n$. Now, $|H_i|=v\log n$, which is a power of $2$ and contains the complete universe. If $v$ and $q$ satisfy $\log q>-2/v$, the TreeSplit algorithm can be applied to $H_i$ with random divisions to ensure that all $v\log n$ users obtain $T$. So for both types of groups, at least $v\log n$ users from each group can obtain the complete universe $T$. Since there are $\alpha n/w\log n$ such groups, we see that at least $\frac{\alpha n}{w\log n}v\log n$ users obtain the complete universe $T$
\end{proof}

We now consider two cases in Sections \ref{case1} and \ref{case2}. 
\subsubsection{Case i: $q$ chosen a-priori} \label{case1}
\begin{lemma}
When $m=\alpha n$ and $q$ is chosen a-priori, the Partition+TreeSplit algorithm has an approximation ratio of $2$.
\end{lemma}
\begin{proof}
We choose $w$ such that $\log q<-2/w$. By Lemma \ref{splittree}, we know that if we can choose a $v$ such that $v\log n$ is a power of $2$ and $\log q>-2/v$, then by the Partition+TreeSplit algorithm, at least $\frac{\alpha n}{w\log n}v\log n$ users obtain the complete universe $T$. So we choose $v$ such that $v\log n$ is the power of $2$ closest to but less than $w\log n$. We also ensure while choosing $w$ that $w\log n$ itself is not a power of $2$. Note that $v/w>1/2$, by definition. Therefore, the fraction of all users who obtain the achievable universe is $\frac{\alpha n}{w\log n}\frac{v\log n}{m}=v/w>1/2$.
\qed
\end{proof}
\subsubsection{Case ii: $q$ can be chosen a-posteriori} \label{case2}
\begin{lemma}
When $m=\alpha n$ and $q$ can be chosen a-posteriori, the Partition+TreeSplit algorithm has an approximation ratio of $1+\epsilon_1$ for some arbitrarily small $\epsilon_1>0$.
\end{lemma}
\begin{proof}
We can now choose all of $w$, $v$ and $q$ as we please such that $-2/v<\log q<-2/w$, and Lemma \ref{splittree} assures us that at least $\frac{\alpha n}{w\log n}v\log n$ users obtain the complete universe $T$. We first choose $v$ such that $v\log n$ is a power of $2$. We then choose $w$ arbitrarily close to but greater than $v$ ($w=v+\delta$, say). We then choose $q$ such that $-2/v<\log q<-2/w$, and this choice is possible because of our choice of $v$ and $w$. Operating with these values of $v$, $w$ and $q$, we see that the fraction of users who finally obtain the complete universe $T$ is $\frac{\alpha n}{w\log n}\frac{v\log n}{m}=v/w\approx 1$.
\qed
\end{proof}
\begin{figure}
\centering
    \psfrag{A}[][][.75]{$m$ users}
    \psfrag{B}[][][.65]{$w\log n$ users}
    \psfrag{C}[][][.75]{Set of file covers $\mathcal{G}$}
    \psfrag{D}[][][.75]{}
    \psfrag{E}[][][.65]{$v\log n$ users}
    \psfrag{F}[][][.75]{$\frac{m}{w\log n}$ groups}
    \psfrag{G}[][][.75]{}
    \includegraphics[scale=.13]{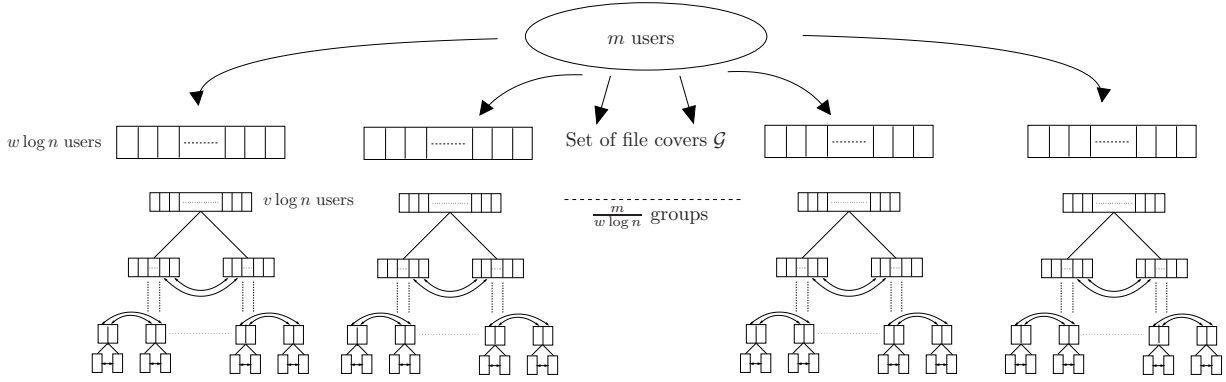}
    \caption{Partition+TreeSplit algorithm when $m=\BigO{n}$ or $\BigO{n^z}$. The users are first randomly partitioned into a set of groups $\mathcal{G}$. Each group has size $w\log n$, and is a file cover with high probability. $v\log n$ users from each group are then chosen by the UniquePick algorithm and the TreeSplit algorithm is applied on them.} \label{treecover}
\end{figure}
\subsection{Approximation ratio of $\frac{1+z}{2}+\epsilon$ for $m=\BigO{n^z}$, $z>1$}
The algorithm for this case is still the Partition+TreeSplit algorithm. We recall that it involves the following: randomly partition the users into groups of size $w\log n$, and then perform the TreeSplit algorithm with random divisions on $v\log n$ users in each group. Consequently, we follow an analysis similar to that of Section \ref{nsection}. 
\begin{lemma}
If $m=\alpha n^z$ and $w$, $v$ and $q$ are such that $-2/v<\log q<-(1+z)/w$ and $v\log n$ is a power of $2$ closest to but less than $w\log n$, then the Partition+TreeSplit algorithm ensures that at least $\frac{\alpha n^z}{w\log n}v\log n$ users obtain the complete universe $T$. \label{splittree'}
\end{lemma}
\begin{proof}
As in Section \ref{nsection}, we want to avoid errors $E_1$ and $E_2$ in step \textbf{(i)} and step \textbf{(ii)} of the algorithm, respectively. In this regime, however, $m=\alpha n^z$. So,
\begin{equation}
\mathcal{P}(E_1\cup E_2)<\frac{\alpha n^z}{w\log n}\Big(nq^{w\log n} + 2v\log n e^{(n^{\frac{v\log q}{2}} -1)\cdot n^{(1+\frac{v\log q}{2})}}\Big). \label{longexp'}
\end{equation}
Splitting the RHS of \eqref{longexp'} into two terms, we see that the first term goes to zero for $\log q<-(1+z)/w$ and that the second term goes to zero for $\log q>-2/v$. So the probability of error goes to zero as $n\rightarrow\infty$ when $-2/v<\log q<-(1+z)/w$.

As for the number of users who finally obtain the complete universe $T$, we follow a logic similar to that of the proof of Lemma \ref{splittree}. It has been omitted for the sake of brevity.
\qed
\end{proof}

We now consider the two cases as before.
\begin{lemma}
When $m=\alpha n^z$ for $z>1$, and $q$ is chosen a-priori, the Partition+TreeSplit algorithm has an approximation ratio of $1+z$.
\end{lemma}
\begin{proof}
We choose $w$ such that $\log q<-(1+z)/w$. We then choose some $v$ such that $v<2w/1+z$ and $v\log n$ is a power of $2$. Note that we can always choose $v$ such that $1/1+z<v/w<2/1+z$. We are also assured that $\log q>-2/v$. Following the arguments that we made in Section \ref{case1} but now with Lemma \ref{splittree'}, we can say that at least $\alpha n^z\cdot \frac{v}{w}>\alpha n^z\cdot \frac{1}{1+z}$ users obtain the universe.
\qed
\end{proof}

\begin{lemma}
When $m=\alpha n^z$ and $q$ can be chosen a-posteriori, the Partition+TreeSplit algorithm has an approximation ratio of $\frac{1+z}{2}+\epsilon_2$ for an arbitrarily small $\epsilon_2>0$.
\end{lemma}
\begin{proof}
We start by choosing $v$ such that $v\log n$ is a power of $2$. We then set $w=\frac{1+z}{2}v+\delta$. Now, we choose $q$ such that $-2/v<\log q<-(1+z)/w$, which we know is possible. So by Lemma \ref{splittree'}, we can now ensure that at least $\alpha n^z\cdot \frac{v}{w}$ users obtain all files in the universe. Since $v/w\approx2/1+z$, the fraction of users who obtain all files is $\frac{1+z}{2}+\epsilon_2$.
\qed
\end{proof}

In the last 3 sections, we have shown that the TreeSplit and Partition+TreeSplit algorithms are effective for a variety of regimes of $m$ and $n$. In the next section, we show that for all $m$ up to the order of $e^{o(n)}$, it is highly probable that there exists a schedule which ensures that at least half the users obtain the complete universe $T$, thereby showing how beneficial the Random Sampling acquisition paradigm is.
\subsection{Existence of a schedule such that at least half the users obtain the universe for any $m=\BigO{e^{o(n)}}$} \label{existencesection}
We show the existence of a schedule for any $m=\BigO{e^{o(n)}}$ by considering the splitting condition in the TreeSplit algorithm. We consider the division of a particular group $G$ of $d$ users, and define $E_G$ as the event that all possible divisions of $G$ into two equal groups do not obey the splitting condition.

Without loss of generality, let $G=\{u_1,u_2,\ldots, u_d\}$. The total number of ways to divide $d$ users in $G$ into $2$ groups with $d/2$ users each is $\binom{d}{d/2}$. Let the two groups of users obtained from the $i^{th}$  division be denoted by $G_{1i}$ and $G_{2i}$, where $1\leq i\leq \binom{d}{d/2}$. Let $\mathcal{G}_1\triangleq\{G_{1i},1\leq i\leq \binom{d}{d/2}\}$ and similarly $\mathcal{G}_2\triangleq\{G_{2i},1\leq i\leq \binom{d}{d/2}\}$.

Without loss of generality, we assume $F_{G_{2i}} \subseteq F_{G_{1i}} \; \forall \; i$. Note that $F_{G} = F_{G_{2i}} \cup F_{G_{1i}} = F_{G_{1i}} \; \forall \; i$. Let $F_G = \{1, 2, \ldots, a, a+1,\ldots, a+b\}$, again without loss of generality.

\begin{lemma}
Let $\mathcal{I} \triangleq \bigcap_{M\in\mathcal{G}_1}M$. If $E_G$ occurs, $|\mathcal{I}| = 1$. \label{culprit}
\end{lemma}
\begin{proof}
We prove this by contradiction. Case \textbf{(i)}: Suppose $|\mathcal{I}| > 1$; then at least two users are constrained to always reside together after division. That does not tally with the fact that we are dividing $G$ in $\binom{d}{d/2}$ ways, which is a contradiction. Case \textbf{(ii)}: Suppose $|\mathcal{I}| = 0$, it would imply that $\exists \; G_{1i_1},G_{1i_2} \in \mathcal{G}_1, s.t \;G_{1i_1} \cap G_{1i_2} = \phi$; in that case, $G$ could have been divided into $G_{1i_1}$ and $G_{1i_2}$ while obeying the splitting condition. So $E_G$ cannot occur.
\qed
\end{proof}

According to Lemma \ref{culprit}, there is one user who is responsible for the event $E_G$, whom we call the \emph{culprit user}. Without loss of generality, let this user be $u_1$. Let $C_{u_1}^0 = \{1,2,\ldots,a\}$ and define $\mathcal{B} \triangleq F_G\backslash C_{u_1}^0=\{a+1,\ldots,a+b\}$.
\begin{lemma}\label{lemma1}
If $E_G$ occurs, any group $G'$ of $d/2-1$ users from $\{u_2,u_3,\ldots,u_{d}\}$ must satisfy $\mathcal{B} \subseteq F_{G'}$.
\end{lemma}
\begin{proof}
Consider any group $G'\subset G$ consisting $d/2-1$ users $\{u_{g_1},u_{g_2},\ldots,u_{g_{d/2-1}}\}$. These $d/2-1$ users along with user $u_1$ will form a group $G_{1i}\in \mathcal{G}_1$ for some $i$, since that is one of the possible divisions. Also note that $F_{G} = F_{G'}\bigcup C_{u1}^{0}$. Therefore $\{a+1,a+2,\ldots, a+b\} \subseteq F_{G'}$.
\qed
\end{proof}

\begin{lemma}\label{minimumcondition}
At least $d/2$ users from $\{u_2,u_3,\ldots,u_{d}\}$ must possess each file $f_i \in \mathcal{B}$ for $E_G$ to occur, and the probability that this happens is less than $(2p)^{\frac{db}{2}}$, where $b=|\mathcal{B}|$.
\end{lemma}
\begin{proof}
We prove the first part of the Lemma by contradiction. Suppose less than $d/2$ users from $\{u_2,u_3,\ldots,u_{d}\}$ possess some file $f_0\in\mathcal{B}$. This implies that at least $d/2-1$ users from $\{u_2,u_3,\ldots,u_{d}\}$ do not possess $f_0$. Without loss of generality, let users $\{u_2,u_3,\ldots,u_{d/2}\}$ not possess $f_0$. These $d/2-1$ users will not satisfy  Lemma \ref{lemma1}, which is a contradiction. Thus, each file $f_i \in \mathcal{B}$ must be possessed by at least $d/2$ users from $\{u_2,u_3,\ldots,u_{d}\}$.

We now come to the second part of the Lemma. Probability that at least $d/2$ users from $\{u_2,u_3,\ldots,u_{d}\}$ possess a file $f_0$ is $\binom{d-1}{d/2}p^{d/2}$. Since each file $f_0 \in \mathcal{B}$ is chosen independently, probability that at least $d/2$ users from $\{u_2,u_3,\ldots,u_{d}\}$ possess all $b \; (=|\mathcal{B}|)$ files is $\Big(\binom{d-1}{d/2}p^{d/2}\Big)^{b} < (2p)^\frac{db}{2}$.
\qed
\end{proof}

\begin{lemma} \label{lemma:invalid_probab}
$\mathcal{P}(E_G)$ is less than $\Big(p+q\big(q'+p'(2p)^{d/2}\big)\Big)^n$, where $q' =  q^{d-1}$ and $p' = 1 -q'$.
\end{lemma}

\begin{proof}
Let us denote the event in which the culprit user $u_1$ has $a$ files and $F_G$ has $a+b$ files by $E_{a,b}$.
We know from Lemma \ref{minimumcondition} that $\mathcal{P}(E_G|E_{a,b})<(2p)^\frac{db}{2}$.
But
\begin{equation} 
\mathcal{P}(E_G)=\sum_{a=0}^n\sum_{b=0}^{n-a}\mathcal{P}(E_G|E_{a,b})\mathcal{P}(E_{a,b}). \label{totalerrorprob}
\end{equation}
Note that the probability that $d-1$ users do not possess a file is $q^{d-1}$, denoted by $q'$, and
\begin{equation}
\mathcal{P}(E_{a,b})=\binom{n}{a}p^aq^{n-a} \times \binom{n-a}{b}(p')^b(q')^{n-a-b}, \label{probabilitymuser}
\end{equation}
Therefore, from \eqref{totalerrorprob} and \eqref{probabilitymuser},
\begin{align}
\mathcal{P}(E_G)<&\sum_{a=0}^{a=n}\;\sum_{b=0}^{b=n-a}\binom{n}{a}p^aq^{n-a}  \binom{n-a}{b}(p')^b(q')^{n-a-b}(2p)^{db/2}, \nonumber\\
=&\sum_{a=0}^{a=n}\binom{n}{a}p^aq^{n-a}\big(q'+p'(2p)^{d/2}\big)^{n-a}, \nonumber\\
=&\Big(p+q\big(q'+p'(2p)^{d/2}\big)\Big)^n. \qquad \qquad \qquad\quad\label{p'eqn} 
\end{align}
\qed
\end{proof}
\begin{theorem}
For $m = \BigO{\exp(n^{1-\epsilon_3})}$ ($\epsilon_3>0$) and $p < 1/2$, there exist divisions obeying the splitting condition for all groups in the TreeSplit algorithm as $ n \rightarrow \infty$.
\end{theorem}
\begin{proof}
Let $\mathcal{K}$ represent the set of $m-1$ groups which are divided in the TreeSplit algorithm. We say that an \emph{error} occurs if $\exists \;G' \in \mathcal{K}$, for which all divisions do not satisfy splitting condition.
\begin{equation}
\mathcal{P}(error) \leq \sum_{G' \in \mathcal{K} }  \mathcal{P}(E_{G'}) \leq |\mathcal{K}| \Big(\max_{G'\in\mathcal{K}} \{\mathcal{P}(E_{G'})\}\Big) \label{maxprobab}
\end{equation}
From equation \eqref{p'eqn}, it is easy to see that $\max\{\mathcal{P}(E_{G'})\}$ occurs $G'$ has minimum number of users, i.e. $d=2$, and is equal to $\Big(p+q\big(q+2p^2\big)\Big)^n$. Also, $|\mathcal{K}|=m-1<m$. Therefore, from equation \eqref{maxprobab}, we obtain
\begin{equation}
\mathcal{P}(error) \leq  \BigO{\exp(n^{1-\epsilon})}\big(p+q(1-p+2p^2)\big)^n \label{lastequation}
\end{equation}
For $p<1/2$, the RHS of equation \eqref{lastequation} tends to zero as $n\rightarrow\infty$. This implies that $\forall \;G' \in \mathcal{K}$ at least one division exists which satisfies the splitting condition. 
\qed
\end{proof}
So the TreeSplit algorithm with some method of division exists with high probability even for $m = \BigO{e^{o(n)}}$, and provides a schedule through which at least half the total number of users obtain the complete universe.
\section{Simulations}
We ran simulations to verify our result from Section \ref{existencesection}. We considered large $m$ when compared to $n$. We divided users as we would in the TreeSplit algorithm, and noted an error whenever any one group in the tree was such that all possible divisions violated the splitting condition. Shown in Fig. \ref{simulation} is a plot of the lowest value of $n$ for a given $m$ for which the percentage of error cases was less than $1\%$. Notice that for $p=0.10$, even if there are $2^{15}$ users in the system and only $160$ files, $99\%$ of cases are such that all divisions obey the splitting condition, showing that it is very probable that there exists a schedule through which $2^{14}$ users obtain all $160$ files even though each user only starts out with around $160\times0.10=16$ files.
\section{Conclusions and Future Work}
We presented a random file acquisition paradigm and showed constant approximation algorithms for a variety of regimes under the give-and-take file-exchange protocol. We also showed that it is beneficial for users to acquire files using the Random Sampling file acquisition paradigm (with small $p$, to keep cost low) by showing that there exists a schedule that ensures at least half the users obtain the universe with high probability. Future work in this area could involve the investigation of other file-acquisition paradigms and comparisons of their performance with respect to the give-and-take protocol. It could also involve an interesting game theoretic question thrown up by the GT criterion, as to what are the minimum number of files each node should download from the server to ensure that it finally gets all the files that it desires.
\begin{figure}
    \begin{center}
    \psfrag{ylabel}[][][.75]{Number of files $n$}
    \psfrag{xlabel}[][t][.75]{$\log_2m$}
    \psfrag{p = 0.05}[][][.6]{$p=0.05$}
    \psfrag{p = 0.1}[][][.6]{$p=0.1$}
    \psfrag{p = 0.2}[][][.6]{$p=0.2$}
    \psfrag{p = 0.3}[][][.6]{$p=0.3$}
    \psfrag{p = 0.45}[][][.6]{$p=0.45$}
    \psfrag{p = 0.5}[][][.6]{$p=0.5$}
    \psfrag{p = 0.6}[][][.6]{$p=0.6$}
    \psfrag{p = 0.9}[][][.6]{$p=0.9$}
    \psfrag{30}[][][.8]{$30$}
    \psfrag{70}[][][.8]{$70$}
    \psfrag{110}[][][.8]{$110$}
    \psfrag{150}[][][.8]{$150$}
    \psfrag{190}[][][.8]{$190$}
    \psfrag{230}[][][.8]{$230$}
    \psfrag{270}[][][.8]{$270$}
    \psfrag{3}[][][.8]{$3$}
    \psfrag{5}[][][.8]{$5$}
    \psfrag{7}[][][.8]{$7$}
    \psfrag{9}[][][.8]{$9$}
    \psfrag{11}[][][.8]{$11$}
    \psfrag{13}[][][.8]{$13$}
    \psfrag{15}[][][.8]{$15$}
    \includegraphics[scale=.7]{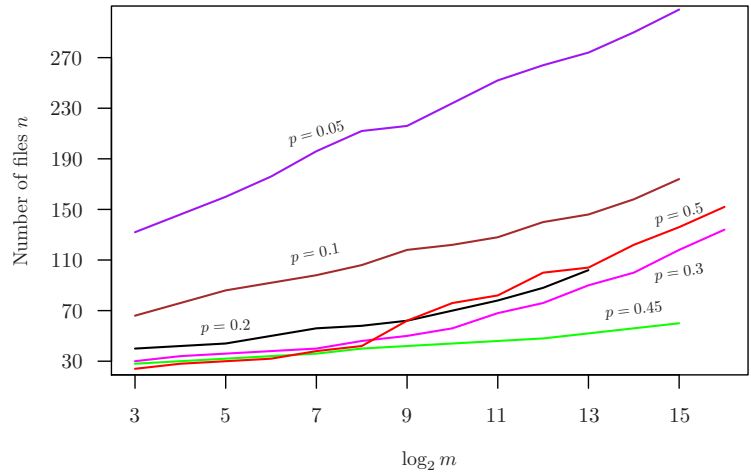}
    \caption{Plot for different values of $p$ of the minimum value of $n$ for which $99\%$ of cases are such that there exist divisions for all groups, with $\log_2m$.} \label{simulation}
    \end{center}
\end{figure}
\bibliographystyle{IEEEtran}
\bibliography{research}
\end{document}